%% file: main.tex
\documentclass{l4dc2024}

\usepackage{amsfonts}
\usepackage{mathtools} % coloneqq
\usepackage{textcomp}
\usepackage{tabu} % Spacing out tables
\usepackage{xcolor}
\usepackage{multirow}
\usepackage{comment}
\usepackage{wrapfig}
\usepackage{titlesec}

\titlespacing*{\section}{0pt}{7pt}{3pt}
\titlespacing*{\subsection}{0pt}{7pt}{3pt}
\titlespacing*{\subsubsection}{0pt}{9pt}{5pt}

\usepackage{caption}
\DeclareRobustCommand{\rchi}{{\mathpalette\irchi\relax}}
\newcommand{\irchi}[2]{\raisebox{\depth}{$#1\chi$}} % inner command, used by \rchi

\newtheorem{prop}{Proposition}
\newtheorem{lem}{Lemma}

\newtheorem{problem}{Problem}
\newtheorem{defn}{Definition}
\newtheorem{rem}{Remark}
\newtheorem{cor}{Corollary}
\newtheorem{assum}{Assumption}

\newcommand{\panos}[1]{\textcolor{red}{#1}}

\graphicspath{{./Figures/}}

\newcommand{\vct}[1]{\boldsymbol{#1}}

\renewcommand{\Re}{\mathbb{R}}
\newcommand{\E}{\mathbb{E}}

\newcommand{\vectorize}{\operatorname{vec}}

\title[]{Data-Driven Robust Covariance Control for Uncertain Linear Systems}
\usepackage{times}

\author{%
\Name{Joshua Pilipovsky} \Email{jpilipovsky3@gatech.edu}\\
\addr Daniel Guggenheim School of Aerospace Engineering, Georgia Institute of Technology
\AND
\Name{Panagiotis Tsiotras} \Email{tsiotras@gatech.edu}\\
\addr Daniel Guggenheim School of Aerospace Engineering, Georgia Institute of Technology%
}

\begin{document}

\maketitle

\begin{abstract}%
	The theory of covariance control and covariance steering (CS) deals with controlling the dispersion of trajectories of a dynamical system, under the implicit assumption that accurate prior knowledge of the system being controlled is available.
	In this work, we consider the problem of steering the distribution of a discrete-time, linear system subject to exogenous disturbances under an \textit{unknown} dynamics model.
	Leveraging concepts from behavioral systems theory, the trajectories of this unknown, noisy system may be (approximately) represented using system data collected through experimentation.
	Using this fact, we formulate a direct data-driven covariance control problem using input-state data.
	We then propose a maximum likelihood uncertainty quantification method to estimate and bound the noise realizations in the data collection process.
	Lastly, we utilize robust convex optimization techniques to solve the resulting norm-bounded uncertain convex program.
	We illustrate the proposed end-to-end data-driven CS algorithm on a double integrator example and showcase the efficacy and accuracy of the proposed method compared to that of model-based methods.
\end{abstract}


\section{Introduction}

Controlling the uncertainty is paramount for the development of safe and reliable systems. Originating from the pioneering contributions of~\cite{HS2}, the theory of covariance control addresses the problem of asymptotically steering the \textit{distribution} of a linear dynamical system from an initial to a terminal distribution when the system dynamics are corrupted by additive disturbances.
When the time horizon is finite, the covariance control problem is often referred to as covariance steering (CS).
This domain has evolved substantially in recent years, expanding to encompass more pragmatic scenarios. These extensions include incorporating probabilistic constraints on the state and the input \citep{Josh_IRA_journal, Bakolas_constrained_CS}, the ability to steer more complex distributions \citep{nonGaussianCS}, and adaptations for receding horizon implementations \citep{saravanos2022distributed} among many others.
Covariance steering has also been successfully applied to diverse contexts such as urban air mobility \citep{exact_CS_2}, vehicle path planning \citep{kazu_PP}, high-performance, aggressive driving \citep{Knaup_CSMPC}, spacecraft rendezvous \citep{Josh_DRCS}, and interplanetary trajectory optimization \citep{JoshJack}.

A common underlying assumption in all the previous methods is the availability of a model of the system dynamics, typically derived from physical first principles.
Acquiring such a model often involves either direct data acquisition or a synthesis of empirical data with first principles, a process broadly categorized under system identification. 
These approaches, while effective, present some notable challenges. 
First principles modeling demands extensive domain-specific knowledge and effort to accurately represent the system. 
System identification offers a balance between accuracy and complexity but introduces complications in subsequent control design phases, such as determining the optimal model for control law synthesis. 
This \textit{indirect} data-driven control approach leads to a complex bi-level optimization problem, involving both model identification and control design, which is generally inseparable \citep{Dorfler_DD_special_issue}. 
Consequently, research efforts have been devoted to the investigation of alternative strategies such as dual control \citep{dual_control} and combined identification for control 
\citep{sysID_for_control}.

Our work deviates from these traditional methodologies by delving into \textit{direct} data-driven control design. 
The proposed approach bypasses the necessity for a parametric system model, instead deriving feedback control laws directly from empirical data. 
By applying principles from behavioral systems theory, namely \textit{Willems' Fundamental Lemma}~\citep{WFL} we can effectively characterize the trajectories of a Linear Time-Invariant (LTI) system using the Hankel matrix of the input/state data \citep{behavioral_system_theory, WFL}.
In cases when measurements are accurate, and in the absence of noise,
this procedure exactly characterizes the system behavior, laying the groundwork for control design based on this data-driven paradigm.
This method has demonstrated its effectiveness, e.g., in solving the Linear Quadratic Regulator (LQR) problem \citep{DataDriven_dePersis_1, DataDriven_dePersis_2} without knowledge of the system matrices.
It has been further adapted to situations with noisy data via regularization strategies~\citep{DataDriven_regularization, DataDriven_CE}. 
These developments successfully bridge the gap between \textit{certainty-equivalence} (CE) and \textit{robust} control paradigms. 
Additionally, recent advancements \citep{DataDriven_Petersen, JP_DDCS_noiseless} have extended this data-driven method to the design of feedback controllers that are robust to various forms of \textit{bounded} disturbances and to uncertainties in the initial state distribution.

This paper introduces a novel approach to data-driven covariance control design in systems subjected to additive, potentially \textit{unbounded}, Gaussian disturbances.
Given that the collected data is influenced by noise, the resulting optimization program incorporates the unknown noise \textit{realization} from the input/output dataset. 
We tackle this problem by first estimating the noise realization sequence and the underlying noise covariance matrix using maximum likelihood methods, while also establishing norm bounds for the estimation error.
We then employ the concept of the robust counterpart of an uncertain convex program to enhance the tractability of the resulting convex optimization program.
We illustrate the proposed framework on a double integrator system and compare the performance and robustness of the data-driven controller with that of a model-based controller.

%---------------------------------------------------%
\section{Problem Statement}~\label{sec:PS}
\vspace{-0.6cm}
\subsection{Notation}~\label{subsec:notation}

Real-valued vectors are denoted by lowercase letters, $u\in\Re^{m}$, matrices are denoted by uppercase letters, $V\in\Re^{n\times m}$, and random vectors are denoted by boldface letters, $\vct w\in\Re^{p}$.
$\rchi_{p,q}^{2}$ denotes the inverse cumulative distribution function (CDF) of the chi-square distribution with $p$ degrees of freedom and quantile $q$.
The Kronocker product is denoted as $\otimes$ and the vectorization of a matrix $A$ is denoted as $\vectorize(A) = [a_1^\intercal, \ldots, a_{M}^\intercal]^\intercal$, where $a_i$ is the $i$th column of $A$.
Given two matrices $A$ and $B$ having the same number of columns, the matrix $[A;B]$ denotes the stacking of the two matrices columnwise.
The set $\mathbb{N}_{[a,b]}$ with $a < b$, denotes the set of natural numbers between $a, b \in\mathbb{N}$.
We denote the two-norm by $\|\cdot\|$ and the matrix Frobenius norm by $\|\cdot\|_{\mathrm{F}}$.
Lastly, we denote a discrete-time signal $z_0, z_1, \ldots, z_{T-1}$ by $\{z_k\}_{k=0}^{T-1}$.
\subsection{Data-Driven Covariance Steering}~\label{subsec:DDCS_problem}

We consider the following discrete-time, stochastic, time-invariant system
\begin{equation}~\label{eq:1}
	\vct x_{k+1} = A \vct x_k + B \vct u_k + D \vct w_k, \quad k\in\mathbb{N}_{[0, N-1]},
\end{equation}
where $\vct x\in\mathbb{R}^{n}, \vct u\in\mathbb{R}^{m}$, and $\vct w\sim\mathcal{N}(0, I_{d})$ for all $k = 0,1,\ldots, N - 1$, where $N$ represents the time horizon.
We assume the noise is i.i.d. at each time step and is uncorrelated with the initial state, that is, $\E[\vct w_k \vct w_j^\intercal] = \E[\vct x_0 \vct w_k^\intercal] = 0$, for all $j\neq k$.
The system matrices $A, B, D$ are assumed to be constant, but unknown.
The initial uncertainty in the system resides in the initial state $\vct x_0$, which is a random $n$-dimensional vector drawn from the normal distribution
\begin{equation}~\label{eq:2}
	\vct x_0 \sim \mathcal{N}(\mu_i, \Sigma_i),
\end{equation}
where $\mu_i\in\mathbb{R}^n$ is the initial state mean and $\Sigma_i\in\mathbb{R}^{n\times n} \succ 0$ is the initial state covariance. 
The objective is to steer the trajectories of (\ref{eq:1}) from the initial distribution (\ref{eq:2}) to the terminal distribution
\begin{equation}~\label{eq:3}
	\vct x_N = \vct x_f \sim \mathcal{N}(\mu_f, \Sigma_f),
\end{equation}
where $\mu_f$ and $\Sigma_f\in\Re^{n\times n} \succ 0$ are the desired state mean and covariance at time $N$, respectively. 
Without loss of generality, we may assume that $\mu_f = 0$.
The cost function to be minimized is 
\begin{equation}~\label{eq:4}
	J(u_0,\ldots,u_{N-1}) \coloneqq \mathbb{E}\bigg[\sum_{k=0}^{N-1} \vct x_k^\intercal Q_k \vct x_k + \vct u_k^\intercal R_k \vct u_k\bigg],
\end{equation}
where $Q_k\succeq 0$ and $R_k \succ 0$ for all $k = 0,\ldots,N-1$. 
Let $\mathbb{D} := \{x_k^{(d)}, u_k^{(d)}, x_{T}^{(d)}\}_{k=0}^{T-1}$ be a dataset collected from the system with an experiment over the time horizon $T$.
In general, $N\neq T$.

\begin{comment}
	\textit{Remark 1}: We assume that the system (\ref{eq:1}) is controllable, that is, for any $x_0,x_f\in\mathbb{R}^n$, and no noise ($w_k\equiv 0,~ k = 0,\ldots,N-1$), there exists a sequence of control inputs $\{u_k\}_{k=0}^{N-1}$ that steer the system from $x_0$ to $x_f$.
\end{comment}

\begin{problem}[Data-Driven CS Controller Design]~\label{problem:1}
	Given the unknown linear system (\ref{eq:1}), find the optimal control sequence $\{\vct u_k\}_{k=0}^{N-1}$ that minimizes the objective function (\ref{eq:4}), subject to the initial (\ref{eq:2}) and terminal  (\ref{eq:3}) state distributions using the dataset $\mathbb{D}$.
\end{problem}
\subsection{Problem Reformulation}~\label{subsec:PS_reformulation}

Borrowing from the work of \cite{exact_CS}, we adopt the control policy
\begin{equation}~\label{eq:controlLaw}
	\vct u_k = K_k(\vct x_k - \mu_k) + v_k,
\end{equation}
where $K_k\in\Re^{m\times n}$ is the feedback gain that controls the covariance of the state and $v_k\in\Re^{m}$ is the feed-forward term that controls the mean of the state.
Under the control law \eqref{eq:controlLaw}, it is possible to re-write Problem~\ref{problem:1} as a convex program, which can be solved to optimality using off-the-shelf solvers~\citep{YALMIP}.
With no  constraints present, the state distribution remains Gaussian at all time steps, completely characterized by its first two moments.
Consequently, we may decompose the system dynamics \eqref{eq:1} into the mean dynamics and covariance dynamics.
Substituting the control law \eqref{eq:controlLaw} into the dynamics \eqref{eq:1} yields the decoupled mean and covariance dynamics
\begin{subequations}~\label{eq:Dynamics}
	\begin{align}
		\mu_{k+1} &= A\mu_k + B v_k, \label{eq:meanDynamics} \\
		\Sigma_{k+1} &= (A + B K_k)\Sigma_k(A + B K_k)^\intercal + DD^\intercal. \label{eq:covDynamics}
	\end{align}
\end{subequations}
In the sequel, and similar to the recent work by \cite{exact_CS_2}, we treat the moments of the intermediate states $\{\Sigma_k, \mu_k\}_{k=0}^{N}$ over the steering horizon as decision variables in the resulting optimization problem.

Similar to the dynamics in \eqref{eq:Dynamics}, the cost function can be decoupled as $J = J_{\mu}(\mu_k, v_k) + J_{\Sigma}(\Sigma_k, K_k),$ where
\begin{subequations}
	\begin{align}
		J_{\mu} &:= \sum_{k=0}^{N-1}\left(\mu_k^\intercal Q_k \mu_k + v_k^\intercal R_k v_k\right), \label{eq:meanCost} \\
		J_{\Sigma} &:= \sum_{k=0}^{N-1}\Big(\textrm{tr}(Q_k\Sigma_k) + \textrm{tr}(R_kK_k\Sigma_k K_k^\intercal) \Big). \label{eq:covCost}
	\end{align}
\end{subequations}
Lastly, the boundary conditions take the form
\begin{subequations}
	\begin{align}
		&\mu_0 = \mu_i, \quad \mu_N = \mu_f, \label{eq:meanBCs} \\
		&\Sigma_0 = \Sigma_i, \quad \Sigma_N = \Sigma_f, \label{eq:covBCs}
	\end{align}
\end{subequations}
where $\Sigma_i, \Sigma_f \succ 0$.
Problem~\ref{problem:1} is now recast as the following two sub-problems.

\begin{problem}[Data-Driven Mean Steering (DD-MS)]~\label{problem:2}
	Given the mean dynamics \eqref{eq:meanDynamics}, find the optimal mean trajectory $\{\mu_k\}_{k=0}^{N}$ and feed-forward control $\{v_k\}_{k=0}^{N-1}$ that minimize the mean cost \eqref{eq:meanCost} subject to the boundary conditions \eqref{eq:meanBCs} using the dataset $\mathbb{D}$.
\end{problem} 

\begin{problem}[Data-Driven Covariance Steering (DD-CS)]~\label{problem:3}
	Given the covariance dynamics \eqref{eq:covDynamics}, find the optimal covariance trajectory $\{\Sigma_k\}_{k=0}^{N}$ and feedback gains $\{K_k\}_{k=0}^{N-1}$ that minimize the covariance cost \eqref{eq:covCost} subject to the boundary conditions \eqref{eq:covBCs} using the dataset $\mathbb{D}$.
\end{problem} 

Note that both the mean and covariance steering problems rely on the system matrices $A$ and $B$ through the system dynamics \eqref{eq:Dynamics}.
Thus, the two problems stated above, are not yet amenable to a data-driven solution.
%
\begin{comment}
\begin{rem}~\label{rem:2}
    Problem~\ref{problem:2} is a standard quadratic program with linear constraints that can be solved analytically given knowledge of the system matrices \cite{Max1}.
    Problem~\ref{problem:3}, however, is a non-linear program due to the cost term $\mathrm{tr}(R_k K_k \Sigma_k K_k^\intercal)$ and the covariance dynamics.
\end{rem}
\end{comment}
%
In the following section, we review the main concepts from behavioral systems theory that will allow us to parametrize the decision variables in Problems~\ref{problem:2} and \ref{problem:3} in terms of the dataset $\mathbb{D}$.

\section{Data-Driven Parameterization}~\label{sec:DataDrivenDesign}

We use the concept of persistence of excitation, along with Willems' Fundamental Lemma~\citep{WFL} to parameterize the decision variables of the control policy.
First, recall the following definitions from~\cite{DataDriven_dePersis_1}.

\begin{defn}~\label{def:1}
	Given a signal $\{z_k\}_{k=0}^{T-1}$ where $z\in\Re^{\sigma}$, its Hankel matrix is given by
	\begin{equation}
		Z_{i,\ell, j} := 
		\begin{bmatrix}
			z_i & z_{i+1} & \ldots & z_{i+j-1} \\
			z_{i+1} & z_{i+2} & \ldots & z_{i + j} \\
			\vdots & \vdots & \ddots & \vdots \\
			z_{i+\ell-1} & z_{i+\ell} & \ldots & z_{i+\ell+j-2}
		\end{bmatrix} \in \Re^{\sigma\ell \times j},
	\end{equation}
	where $i\in\mathbb{N}_0$ 
 and $\ell,j\in\mathbb{N}$.
	For shorthand of notation, if $\ell=1$, we will denote the Hankel matrix by
	\begin{equation}
		Z_{i,1,j} \equiv Z_{i,j} := [z_i \ z_{i+1} \ \ldots \ z_{i+j-1}].
	\end{equation}
\end{defn}

\begin{defn}~\label{def:2}
	The signal $\{z_k\}_{k=0}^{T-1}: [0,T-1]\cap\mathbb{Z}\rightarrow\Re^{\sigma}$ is \textit{persistently exciting} of order $\ell$ if the matrix $Z_{0,\ell,T-\ell + 1}$ has rank $\sigma\ell$.
\end{defn}

\begin{comment}
	\begin{cor}~\label{rem:3}
		In order for a signal to  be persistently exciting of order $\ell$, it must be sufficiently long, i.e., it must hold that $T \geq (\sigma + 1)\ell - 1$.
	\end{cor} 
\end{comment}

Given the dataset $\mathbb{D}$, let the corresponding Hankel matrices for the input sequence, state sequence, and shifted state sequence (with $\ell = 1$) be 
$U_{0,T} := [u_0^{(d)} \ u_1^{(d)} \ \ldots \ u_{T-1}^{(d)}]$, \ $X_{0,T} := [x_0^{(d)} \ x_1^{(d)} \ \ldots \ x_{T-1}^{(d)}]$, and $X_{1,T} := [x_1^{(d)} \ x_2^{(d)} \ \ldots \ x_{T}^{(d)}]$.
The next result characterizes the rank of the stacked Hankel matrices of the input and output data, and is central to our approach to formulate a tractable data-driven covariance steering problem.

\begin{assum}~\label{assum_rank}
	We assume that the data is persistently exciting, i.e., the Hankel matrix of input/state data is full row rank
	\begin{equation}~\label{eq:rank_condition}
		\mathrm{rank}
		\begin{bmatrix}
			U_{0,T} \\
			X_{0,T}
		\end{bmatrix}
		= n + m.
	\end{equation}
\end{assum}

\begin{rem}~\label{rem:4}
    In this work, the full rank condition in \eqref{eq:rank_condition} is an assumption, rather than a derived condition based on notions of persistency of excitation.
    Indeed, when the system is purely deterministic, \eqref{eq:rank_condition} holds if and only if the input $\{u_k\}_{k=0}^{T-1}$ is persistently exciting of order $n + 1$, a result now called the fundamental Lemma \cite{WFL}.
    However, in the present case, it can be shown that if in addition the disturbance realization $W_{0,T} \triangleq [\vct w_0^{(d)}, \ldots, \vct w_{T-1}^{(d)}]$ is persistently exciting of order $n + 1$, then Assumption~\ref{assum_rank} holds \cite{DataDriven_Petersen}.
\end{rem} 

Assumption~\ref{assum_rank} implies that \textit{any} arbitrary input-state sequence can be expressed as a linear combination of the collected input-state data. 
Furthermore, as shown in the next section, this idea has been used by \cite{DataDriven_dePersis_1} to parameterize arbitrary feedback interconnections as well.
In the following section, based on the work of \cite{DataDriven_dePersis_2}, we parameterize the feedback gains in terms of the input-state data and reformulate Problem~\ref{problem:3} as a semi-definite program (SDP).

\subsection{Direct Data-Driven Covariance Steering}~\label{sec:directCS}

Assuming the signal $\{u_k\}_{k=0}^{T-1}$ is persistently exciting of order $n + 1$, we can express the feedback gains as follows
\begin{equation}~\label{eq:WFL_gains}
	\begin{bmatrix}
		K_k \\
		I_n
	\end{bmatrix}
	= 
	\begin{bmatrix}
		U_{0,T} \\
		X_{0,T}
	\end{bmatrix}
	G_k, \quad k = 0,\ldots, N - 1,
\end{equation}
where $G_k \in \Re^{T\times n}$ are newly defined decision variables that provide the link between the feedback gains and the input-state data.
\begin{theorem}~\label{thm:relaxed_DDCS}
	Using the data-driven parameterization of the feedback gains \eqref{eq:WFL_gains}, Problem~\ref{problem:3} may be relaxed as the convex program
    \begin{subequations}~\label{eq:convexProblem}
        \begin{align}
            &\hspace*{-4mm} \min_{\Sigma_k, S_k Y_k} \bar{J}_{\Sigma} = \sum_{k=0}^{N-1}\left(\mathrm{tr}(Q_k\Sigma_k) + \mathrm{tr}(R_kU_{0,T}Y_kU_{0,T}^\intercal)\right), \label{eq:convexProblem_cost}
        \end{align}
        such that, for all $k = 0,\ldots, N - 1$,
        \begin{align}
            &C_k^{(1)} := 
            \begin{bmatrix}
                \Sigma_k & S_k^\intercal \\
                S_k & Y_k
            \end{bmatrix} \succeq 0, \label{eq:convexProblem_ineqConstraint} \\
            &C_k^{(2)} := 
            \begin{bmatrix}
                \Sigma_{k+1} - \Sigma_{\vct\xi} & (X_{1,T} - \Xi_{0,T})S_k \\
                S_k^\intercal (X_{1,T} - \Xi_{0,T})^\intercal & \Sigma_{k}
            \end{bmatrix} \succeq 0, \label{eq:convexProblem_ineqConstraint2} \\
            &G_k^{(1)} := \Sigma_k - X_{0,T} S_k = 0, , \quad G_k^{(2)} := \Sigma_{N} - \Sigma_{f} = 0, \label{eq:convexProblem_eqConstraints}
        \end{align}
    \end{subequations}
    where $\Xi_{0,T} := [\vct \xi_0, \ldots, \vct \xi_{T-1}]\in\Re^{n\times T}$ is the Hankel matrix of the (unknown) disturbances, and $\vct\xi_k \sim \mathcal{N}(0,\Sigma_{\vct\xi})$, 
    where $\Sigma_{\vct\xi} := DD^\intercal$ for all $k=0,\ldots,T-1$.
\end{theorem}

\begin{proof}
Please see Appendix~A \citep{JP_RDDCS_arxiv}.
\end{proof}

\subsection{Indirect Data-Driven Mean Steering}

Given the mean dynamics \eqref{eq:meanDynamics} in terms of the open-loop control $v_k$, Assumption~\ref{assum_rank} can also be used to provide a system identification type-of-result using the following theorem.
\begin{theorem}~\label{theorem:1}
	Suppose the signal $u_k^{(d)}$ 
 is persistently exciting of order $n + 1$.
	Then, system \eqref{eq:meanDynamics} has the following equivalent representation
	\begin{equation} ~\label{eq:mut+1}
		\mu_{k+1} = (X_{1,T} - \Xi_{0,T})
		\begin{bmatrix}
			U_{0,T} \\
			X_{0,T}
		\end{bmatrix}^{\dagger}
		\begin{bmatrix}
			v_k \\
			\mu_k
		\end{bmatrix}.
	\end{equation}
\end{theorem}
\begin{proof}
	See \cite{DataDriven_dePersis_1} for details.
\end{proof}
%
\begin{comment}
\begin{rem}~\label{rem:nominal_representation}
Theorem~\ref{theorem:1} provides a data-based open-loop representation of a linear system.
Assuming exact knowledge of the noise realization $\Xi_{0,T}$, one may equivalently interpret equation~\eqref{eq:mut+1}
as the solution to the least-squares problem
\begin{equation}
    \min_{B,A} \left\|X_{1,T} - \Xi_{0,T} - [B \ \ A]
    \begin{bmatrix}
        U_{0,T} \\
        X_{0,T}
    \end{bmatrix}
    \right\|_{F},
\end{equation}
where $\|\cdot\|_{F}$ is the Frobenius norm.
\end{rem}
%
\end{comment}

Using Theorem~\ref{theorem:1}, we can express Problem~2 as the following convex problem
\begin{subequations}~\label{eq:meanProblem}
	\begin{equation}
		\min_{\mu_k, v_k} J_{\mu} = \sum_{k=0}^{N-1} (\mu_k^\intercal Q_k \mu_k + v_k^\intercal R_k v_k), \label{eq:meanProblem_cost}
	\end{equation}
 such that, for all $k = 0,\ldots, N - 1,$
	\begin{align}
		F_{\mu}(\Xi_{0,T}) \mu_k + F_{v}(\Xi_{0,T}) v_k - \mu_{k+1} = 0, \quad
		\mu_{N} - \mu_{f} = 0, \label{eq:meanProblem_eqConstraint2}
	\end{align}
\end{subequations}
where $F_{\mu}\in\Re^{n\times n}$ and $F_{v}\in\Re^{n\times m}$ result from the partition of
\begin{equation}
	F := (X_{1,T} - \Xi_{0,T})
	\begin{bmatrix}
		U_{0,T} \\
		X_{0,T}
	\end{bmatrix}^{\dagger} = 
	\begin{bmatrix}
		F_{v}(\Xi_{0,T}) & F_{\mu}(\Xi_{0,T})
	\end{bmatrix}.
\end{equation}
Given knowledge of $\Xi_{0,T}$, the solution of \eqref{eq:meanProblem} yields an indirect design for the feed-forward control that solves the mean steering problem.
In \citep{JP_DDCS_noiseless} it was shown that, in the absence of noise, the previous optimization problem can be solved as a convex program.

\section{Uncertainty Quantification}

\subsection{Maximum Likelihood Estimation of Noise Realization}

The optimization problem \eqref{eq:convexProblem} as well as \eqref{eq:meanProblem} depend on the unknown noise realization.
In this section, we propose a maximum likelihood estimation (MLE) scheme to estimate the noise realization from the collected data.
First, we encode the stochastic linear system dynamics by enforcing a \textit{consistency} condition on the realization data.
To this end, we first notice that there exists a noise realization such that the dynamics \eqref{eq:1} satisfy
\begin{equation}~\label{eq:dynamics_realization}
	X_{1,T} = AX_{0,T} + B U_{0,T} + \Xi_{0,T}.
\end{equation}
For notational convenience, define the augmented Hankel matrix
$S := [U_{0,T}; X_{0,T}]\in\Re^{(m+n)\times T}$,
from which we may re-write \eqref{eq:dynamics_realization} as $X_{1,T} = [B \ A]S + \Xi_{0,T}$.
Additionally, noting that the pseudoinverse satisfies the property $SS^\dagger S = S$, \eqref{eq:dynamics_realization} is equivalently written as
$
	X_{1,T} - \Xi_{0,T} = [B \ A] SS^\dagger S,
$
which, using the 
relation $X_{1,T} - \Xi_{0,T} = [B \ A]S$, yields
\begin{equation}~\label{eq:consistency_equation}
	(X_{1,T} - \Xi_{0,T})(I_{T} - S^\dagger S) = 0.
\end{equation}
Equation \eqref{eq:consistency_equation} is a \textit{model-free} condition that must be satisfied for all noisy linear system data realizations, and hence is a consistency relation for any feasible set of data.
Given the constraint \eqref{eq:consistency_equation}, the MLE  problem then becomes
\begin{equation}~\label{eq:ML_problem}
    \max_{\Xi_{0,T}, \Sigma_{\vct\xi}} \ \mathcal{J}_{\mathrm{MLE}}(\Xi_{0,T}, \Sigma_{\vct\xi} \ | \ \mathbb{D}) \triangleq \sum_{k=0}^{T - 1} \log \rho_{\vct\xi}(\xi_k), \quad \mathrm{s.t.} \quad (X_{1,T} - \Xi_{0,T})(I_{T} - S^\dagger S) = 0,
\end{equation}
where $\rho_{\vct\xi}(x)$ is the probability density function (PDF) of the noise random vector $\vct\xi$.
%given as
%\begin{equation}~\label{eq:Gaussian}
%	\rho_{\vct\xi}(x) = \frac{1}{(2\pi)^{n/2}}(\det\Sigma_{\vct\xi})^{-1/2}\exp\left(-\frac{1}{2}x^\intercal\Sigma_{\vct\xi}^{-1}x\right).
%\end{equation}

\begin{comment}
	\begin{rem}
		Since the dynamics \eqref{eq:1} are uncertain, this implies that there may be \textit{multiple} noise realizations $\Xi_{0,T}$ that satisfy the linear dynamics.
		As such, the purpose of the constrained ML estimation \eqref{eq:ML_problem} is to find the \textit{most likely} sequence of noise realizations, $\hat{\Xi}_{0,T} = \mathrm{argmin}_{\Xi_{0,T}} \ \eqref{eq:ML_problem_cost}$, given that the noise is normally distributed according to \eqref{eq:Gaussian}.
	\end{rem}
\end{comment}

\begin{theorem}~\label{thm:ML_to_DC}
	Assuming $\Sigma_{\vct\xi} \succ 0$, the MLE optimization problem~\eqref{eq:ML_problem} may be solved as the following difference-of-convex (DC) program
	\begin{subequations}~\label{eq:ML_problem_DC}
		\begin{align}
			&\min_{\Xi_{0,T}, \Sigma_{\vct\xi}, U} \ \left[\frac{1}{2}\mathrm{tr}(U) - \left(-\frac{T}{2}\log\det\Sigma_{\vct\xi}\right)\right] \\
			&\quad \mathrm{s.t.} \quad\quad (X_{1,T} - \Xi_{0,T})(I_{T} - S^\dagger S) = 0, 
   \qquad 
                \begin{bmatrix}
				\Sigma_{\vct\xi} & \Xi_{0,T} \\
				\Xi_{0,T}^\intercal & U
			\end{bmatrix} \succeq 0.
		\end{align}
	\end{subequations}
\end{theorem}
%
\begin{comment}
Thus, the relaxed ML program \eqref{eq:ML_problem_DC} can now be viewed as a difference-of-convex (DC) program, since the objective is a difference of convex functions, and the constraints are all convex.
\end{comment}
The previous DC program may be solved using a successive convexification procedure known as the convex-concave procedure (CCP), which is guaranteed to converge to a feasible point 
\citep{CCP_DC_programming, CCP_Boyd}.
%
\begin{comment}
	A brief overview of DC programs and the CCP are outlined in Appendix~A for reference.
\end{comment}
%
\begin{comment}
	\begin{rem}
		The solution of the DC program in \eqref{eq:ML_problem_DC} is contingent on $\Sigma_{\vct\xi} \succ 0$. 
		If $\Sigma_{\vct\xi}$ is singular, however, then $\det\Sigma_{\vct\xi} = 0$, and $\log\det\Sigma_{\vct\xi}$ is undefined, hence the program is infeasible.
		As a result, other techniques, such as NN estimation, should be used in these cases.
		It should also be noted that such cases arise when the number of disturbance channels is \textit{less} than the number of states channels, i.e., $D\in\Re^{n\times d}$, with $d < n$.
	\end{rem}
\end{comment}
%
Oftentimes, the disturbance matrix $D$ is known beforehand, as is the case when one knows how the disturbances affect the system state variables.
%
\begin{comment}
For example, in a constant-velocity, double integrator system with acceleration noise, the disturbance matrix is given by $D = [0_3, \frac{\Delta T^2}{2}I_{3}]^\intercal$.
\end{comment}
%
In such cases, $\Sigma_{\vct\xi} = DD^\intercal$ is no longer a decision variable, and the MLE problem \eqref{eq:ML_problem} may be solved analytically as given by the following Corollary~\ref{cor:ML_solution}.
\begin{cor}~\label{cor:ML_solution}
	Suppose $\Sigma_{\vct\xi}$ is known. 
	Then the MLE problem \eqref{eq:ML_problem} has the exact solution
	\begin{equation}~\label{eq:ML_estimation}
		\Xi_{0,T}^{\star} = X_{1,T}(I_{T} - S^\dagger S).
	\end{equation}
\end{cor}
%
\begin{comment}
	\begin{rem}
		Notice that the optimal ML noise realization $\hat{\Xi}_{0,T}$ from \ref{cor:ML_solution} is \textit{independent} of the covariance of the noise $\Sigma_{\vct\xi}$.
		Thus, one tractable method to obtain $\hat{\Sigma}_{\vct\xi}$ is to compute the sample covariance from the $\hat{\Xi}_{0,T}$, that is,
		\begin{equation}~\label{eq:sample_covariance}
			\hat{\Sigma}_{\vct\xi} = \frac{1}{T} \sum_{k=0}^{T-1}\hat{\xi}_k\hat{\xi}_k^\intercal = \frac{1}{T}\hat{\Xi}_{0,T}\hat{\Xi}_{0,T}^\intercal.
		\end{equation}
	\end{rem}
\end{comment}

\begin{comment}
	\begin{rem}
		The covariance dynamics \eqref{eq:convexProblem_eqConstraint1} only requires knowledge of $\hat{\Sigma}_{\vct\xi}$ and $\hat{\Xi}_{0,T}$.
		However, the disturbance matrix may be found from $\hat{\Sigma}_{\vct\xi}$ either by taking the Cholesky decomposition $\hat{D} = \hat{\Sigma}_{\vct\xi}^{1/2}$, or, in the case where $D$ is \textit{non}-square, by solving the optimization problem
		\begin{equation}
			\min_{D} \ \|\hat{\Sigma}_{\vct\xi} - DD^\intercal\|_{F},
		\end{equation}
		where $\|\cdot\|_{F}$ is the Frobenius norm.
	\end{rem}
\end{comment}
%

\subsection{Uncertainty Error Bounds}~\label{eq:subsec:error_bounds}
%\vspace{-0.1cm}
We are interested in deriving bounds to ensure robust satisfaction (with high probability) of the terminal CS constraints \eqref{eq:convexProblem_ineqConstraint2} for the DD-CS problem.
To this end, we use the statistical properties of $\Xi_{0,T}$ and generate an ellipsoidal uncertainty set based on some degree of confidence $\delta\in[0.5, 1)$.
For simplicity, assume $\Sigma_{\vct\xi} \succ 0$ is known.
First, we re-write the MLE problem \eqref{eq:ML_problem} in terms of the vectorized parameters to be estimated $\xi := \vectorize(\Xi_{0,T}) = [\xi_0^\intercal,\ldots,\xi_{T-1}^\intercal]^\intercal\in\Re^{nT}$ as
\begin{equation}
	\min_{\xi} \ \mathcal{J}_{\mathrm{MLE}}(\xi \ | \ \mathbb{D}) = \frac{1}{2} \xi^\intercal (I_{T}\otimes \Sigma_{\vct\xi}^{-1})\xi \quad \mathrm{s.t.} \quad C(\xi) := (\Gamma \otimes I_{n})\xi - \lambda = 0, \label{eq:vectorized_ML}
\end{equation}
where $\Gamma := I_{T} - S^\dagger S\in\Re^{T\times T}, \Lambda := X_{1,T} \Gamma \in\Re^{n\times T}$, and $\lambda = \vectorize (\Lambda)$.
It can be shown that \citep{MLE_statistics}, as the number of samples grows, the MLE noise estimate $\hat{\xi}$ converges to a normal distribution as
$\sqrt{T} (\xi - \hat{\xi}) \xrightarrow{d} \mathcal{N}(0, \mathcal{I}^{-1})$, where $\mathcal{I} = \E_{\xi}\left[\frac{\partial^2}{\partial\xi^2}\mathcal{J}_{\mathrm{MLE}}(\xi \ | \ \mathbb{D})\right]$
is the \textit{Fisher Information Matrix} (FIM), which is given by $\mathcal{I} = I_{T}\otimes \Sigma_{\vct\xi}^{-1}$ in the unconstrained case.
For a constrained MLE problem, it can similarly be shown, as by \cite{Cramer_Rao_Bound}, that the asymptotic distribution of the estimates has covariance $\Sigma_{\Delta} = \mathcal{I}^{-1} - \mathcal{I}^{-1}J^\intercal (J\mathcal{I}^{-1}J^\intercal)^{-1}J\mathcal{I}^{-1}$, where $J := \frac{\partial}{\partial\xi}C(\xi) = \Gamma \otimes I_{n}$ is the Jacobian of the constraints.
We next present a few technical results that are useful for computing the uncertainty covariance of the MLE estimates and the associated confidence sets.
The proofs of these results are given in the Appendix \citep{JP_RDDCS_arxiv}.

\begin{lem}~\label{lemma:MLE_convergence}
	The error of the constrained MLE estimator \eqref{eq:ML_problem} for the unknown noise realization $\xi = \vectorize (\Xi_{0,T})$ converges to the normal distribution $\mathcal{N}(0, \Sigma_{\Delta})$, where $\Sigma_{\Delta} = S^\dagger S \otimes \Sigma_{\vct \xi}$.
\end{lem}

\begin{prop}~\label{prop:ambiguity_sets}
	Assume the uncertainty error estimate is normally distributed as $\Delta \xi \sim\mathcal{N}(0,\Sigma_{\Delta})$.
 Then, given some level of risk $\delta\in[0.5, 1)$, the uncertainty set $\Delta := \{\|\Delta\Xi_{0,T}\| \leq \rho\}$,	contains the $(1-\delta)$-quantile of $\Delta\Xi_{0,T}$, with 
    $
        \rho = \frac{\rchi_{nT, 1 - \delta}}{\sqrt{\lambda_{\min}(\Sigma_{\Delta}^{-1})}}
    $,
    where $\rchi_{p,q}$ is the square root of the inverse CDF of the $\rchi_{p,q}^2$ distribution.
\end{prop}
\begin{cor}~\label{cor:noise_error_bound}
	For the MLE problem \eqref{eq:ML_problem}, the associated $(1-\delta)$-quantile uncertainty set
 $\Delta := \{\|\Delta\Xi_{0,T}\| \leq \rho\}$
 has the bound $\rho = \|\Sigma_{\vct\xi}^{1/2}\| \rchi_{nT,1-\delta}$, where $\Delta\Xi_{0,T} = \Xi_{0,T} - \hat{\Xi}_{0,T}$.
\end{cor}

\begin{comment}
	\begin{prop}
		If the uncertainty error estimation is normally distributed as $\Delta \xi_k \sim\mathcal{N}(0,\Sigma_{\Delta_k})$, then, given some level of risk $\delta\in[0.5, 1)$, the uncertainty set $\Delta := \{\|\Delta\Xi_{0,T}\| \leq \rho\}$,	contains the $(1-\delta)$-quantile of $\Delta\Xi_{0,T}$, where 
		\begin{equation}
			\rho := \sqrt{T}\max_{k=0,\ldots,T-1}\|\Sigma_{\Delta_k}^{1/2}\| \rchi_{n,1-\delta},
		\end{equation}
		where $\rchi_{p,q}$ is the square root $q$-quantile of the chi-squared distribution with $p$ degrees of freedom.
	\end{prop}
	%
	\begin{rem}
		In practice, one can pre-compute the upper bound $\rho$ from the data matrices as follows.
		The covariance of the error estimates are given by $\Sigma_{\Delta_k} = \gamma_{kk}^{-2}\Sigma_{\vct\xi}$ and thus $\|\Sigma_{\Delta_k}^{1/2}\| = \gamma_{kk}^{-1}\|\Sigma_{\vct\xi}^{1/2}\|$.
		Letting $\mu := \min_{k=0,\ldots,T-1} \mathrm{diag}(I_{T} - S^\dagger S)$, the upper bound may be computed as
		\begin{equation}
			\rho = \sqrt{T} \mu^{-1}\rchi_{n,1-\delta} \|\Sigma_{\vct\xi}^{1/2}\|.
		\end{equation}
	\end{rem}
\end{comment}

\section{Robust DD-CS}~\label{sec:RDDCS}

%\subsection{Fundamentals}~\label{subsec:RDDCS_fund}
\begin{comment}
Consider the uncertain convex program
\begin{equation}
\inf_{x\in\Re^{n}} \ f(x) \quad \mathrm{s.t.} \quad g(x, \xi) \leq 0, \ \forall \xi\in\Xi, \label{eq:uncertain_program}
\end{equation}
where the functions $f(x), \ g(x)$ are convex, and $\Xi\subset\Re^{m}$ is some bounded uncertainty set.
This uncertain program implies the optimal minimizer $x^\star$ should satisfy the convex constraints \textit{for all} realizations of the uncertain parameter $\xi$.
Such problems, are, in general intractable to solve directly due to this uncertainty.
However, one may alternatively solve a tractable conservative relaxation of the problem under the worst-case scenario of the uncertainty.
\end{comment}

\vspace*{-0.5cm}
The problem outlined in \eqref{eq:convexProblem} is categorized as an uncertain convex program~\citep{robust_optimization}, where the LMI constraints \eqref{eq:convexProblem_ineqConstraint2} are required to hold across all realizations of the uncertain parameter $\Delta\Xi_{0,T}$ within the set $\Delta$ defined in Proposition~\ref{prop:ambiguity_sets}.
This formulation results in a semi-infinite problem, making the original constraints intractable. 
To address this, we focus on the robust counterpart (RC) of a class of uncertain LMIs (uLMI), which simplifies the problem to a robust feasibility problem.  
Letting $\mathcal{A}(Z_k, \Xi_{0,T}) \succeq 0$ denote the LMI in \eqref{eq:convexProblem_ineqConstraint2} for the decision variables $Z_k = \{S_k, \Sigma_k,\Sigma_{k+1}\}$, we can ensure tractability by enforcing the condition $\sup_{\Delta\Xi_{0,T}\in\Delta} \mathcal{A}(Z_k, \Xi_{0,T}) \succeq 0$.
Adhering to this robust counterpart ensures that the original constraints are robustly met with high probability, characterized by the risk value $\delta$.
In what follows, we will form the RC of the uncertain CS problem in \eqref{eq:convexProblem}.
We first state the following theorem on the equivalence of the RC of a uLMI in which the uncertainty appears linearly in the constraints.

\begin{comment}
This motivates the following definition.
\begin{defn}
	The robust counterpart (RC) of the uncertain convex program \eqref{eq:uncertain_program} is the program
	\begin{equation*}
		\inf_{x\in\Re^{n}} \ f(x) \quad \mathrm{s.t.} \quad \sup_{\xi\in\Xi} \ g(x,\xi) \leq 0. 
	\end{equation*} 
\end{defn}

	The interpretation is that if the optimizer satisfies the constraints under the worst-case conditions, then surely it will satisfy the constraints for smaller function values as well.
	One can think of the RC of an uncertain program as a game with nature - that of minimizing the objective subject to the most adversarial actions.
\end{comment}

\begin{comment}
%and subsequently make it tractable.
To this end, 
%first we note two important aspects of the uncertain program: The uncertainties themselves lie in a matrix space, that is, 
we will assume that the uncertainty set is given by
%\begin{equation}
$
    \mathcal{Z} := \{\Xi_{0,T} = \hat{\Xi}_{0,T} + \Delta\Xi_{0,T} : \|\Delta\Xi_{0,T}\| \leq \rho\}, %\label{eq:uncertainty_set}
$
%\end{equation}
where the estimate of the uncertainty is computed from \eqref{eq:ML_problem_DC}, and the estimation error uncertainty set is similarly computed from Corollary~\ref{cor:noise_error_bound}.
Additionally, the uncertainties enter through the LMI constraints \eqref{eq:convexProblem_ineqConstraint2}, that is, we seek to enforce the constraints
\begin{equation}
	\mathcal{A}(Z, \Xi_{0,T}) \succeq 0, \quad \Xi_{0,T} \in \mathcal{Z},
\end{equation}
for the decision variables $Z$ and uncertainty set $\mathcal{Z}$.
\end{comment}

\begin{theorem}[\cite{robust_optimization}, Proposition 6.4.1]~\label{thm:uncertain_LMI_RC}
	The RC of the uncertain LMI
	\begin{equation}
		\mathcal{A}(y, \Pi) := \hat{\mathcal{A}}(y) + L^\intercal(y)\Pi R + R^\intercal \Pi^\intercal L(y) \succeq 0, \label{eq:uncertain_linear_LMI}
	\end{equation}
	with unstructured norm-bounded uncertainty set $\mathcal{Z} = \{\Pi \in\Re^{p\times q} : \|\Pi \| \leq \rho\}$, can be  equivalently represented by the LMI
	\begin{equation}
		\begin{bmatrix}
			\lambda I_{p} & \rho L(y) \\
			\rho L^\intercal(y) & \hat{\mathcal{A}}(y) - \lambda R^\intercal R
		\end{bmatrix} \succeq 0, \label{eq:tractable_RC}
	\end{equation}
	in the decision variables $y,\lambda$.
\end{theorem}

\subsection{Problem Formulation and Robust LMI Counterpart}~\label{subsec:RDDS_formulation}

%We now proceed to reformulate the uncertain CS program \eqref{eq:convexProblem} to a more tractable RC using the results of~\cite{xxx}.
%in Theorem~\ref{thm:uncertain_LMI_RC}.
Using \eqref{eq:ML_problem_DC} to obtain a noise realization estimate $\hat{\Xi}_{0,T}$ and Corollary~\ref{cor:noise_error_bound} to obtain the associated uncertainty set $\Delta$,
the original covariance LMI constraints \eqref{eq:convexProblem_ineqConstraint2}, may be decomposed using $\Xi_{0,T} = \hat{\Xi}_{0,T} + \Delta\Xi_{0,T}$ as the semi-infinite uLMIs
%\begin{equation*}
$
    \hat{G}_k^{\Sigma} + \Delta G_k^{\Sigma}(\Delta\Xi_{0,T}) \succeq 0$, for all $\|\Delta\Xi_{0,T}\| \leq \rho,
$
%\end{equation*}
where,
\begin{equation}
	\hat{G}_{k}^\Sigma = 
	\begin{bmatrix}
		\Sigma_{k+1} - \Sigma_{\vct\xi} & (X_{1,T} - \hat{\Xi}_{0,T})S_k \\
		S_k^\intercal(X_{1,T} - \hat{\Xi}_{0,T})^\intercal & \Sigma_{k} \label{eq:nominal_matrix}
	\end{bmatrix},
\end{equation}
is the nominal LMI, and
\begin{equation}
	\Delta G_k^{\Sigma} = 
	\begin{bmatrix}
		0_{n} & -\Delta \Xi_{0,T} S_k \\
		-S_k^\intercal \Delta \Xi_{0,T}^\intercal & 0_{n} 
	\end{bmatrix} \succeq 0, \label{eq:pertburation_matrix}
\end{equation}
is the perturbation to the covariance LMI.
Next, we represent the matrix \eqref{eq:pertburation_matrix} as
\begin{equation}
	\Delta G_k^{\Sigma} = L^\intercal (S_k) \Delta \Xi_{0,T}^\intercal R + R^\intercal \Delta \Xi_{0,T} L(S_k), \label{eq:uncertain_LMI_decomposition}
\end{equation}
where,
$
    %\Psi_{0,T} := \Xi_{0,T}^\intercal, \ 
    L^\intercal (S_k) = [0_{n,T}; -S_k^\intercal], \ R^\intercal = [I_{n}; 0_{n}].
$
Finally, using Theorem~\ref{thm:uncertain_LMI_RC}, we may equivalently represent the RC of the uLMI \eqref{eq:uncertain_LMI_decomposition} as the LMI
\begin{equation}
	\begin{bmatrix}
		\lambda I_{T} & \rho L(S_k) \\
		\rho L^\intercal (S_k) & \hat{G}_k^{\Sigma}(Z_k, S_k) - \lambda R^\intercal R 
	\end{bmatrix} \succeq 0,
\end{equation}
in terms of the decision variables $\{S_k, \lambda\}$ and $Z_k = \{\Sigma_k, \Sigma_{k+1}\}$.
\begin{figure}[!htb]
	\centering
	\hspace*{-0.3cm}
	\subfigure[MB-CS (top) and R/DD-CS (bottom) for nominal dynamics model $\{A, B\}$.]{%
		\label{fig:1a}%
		\includegraphics[trim={1.5cm 3cm 2.3cm 2.1cm}, clip, scale=0.36]{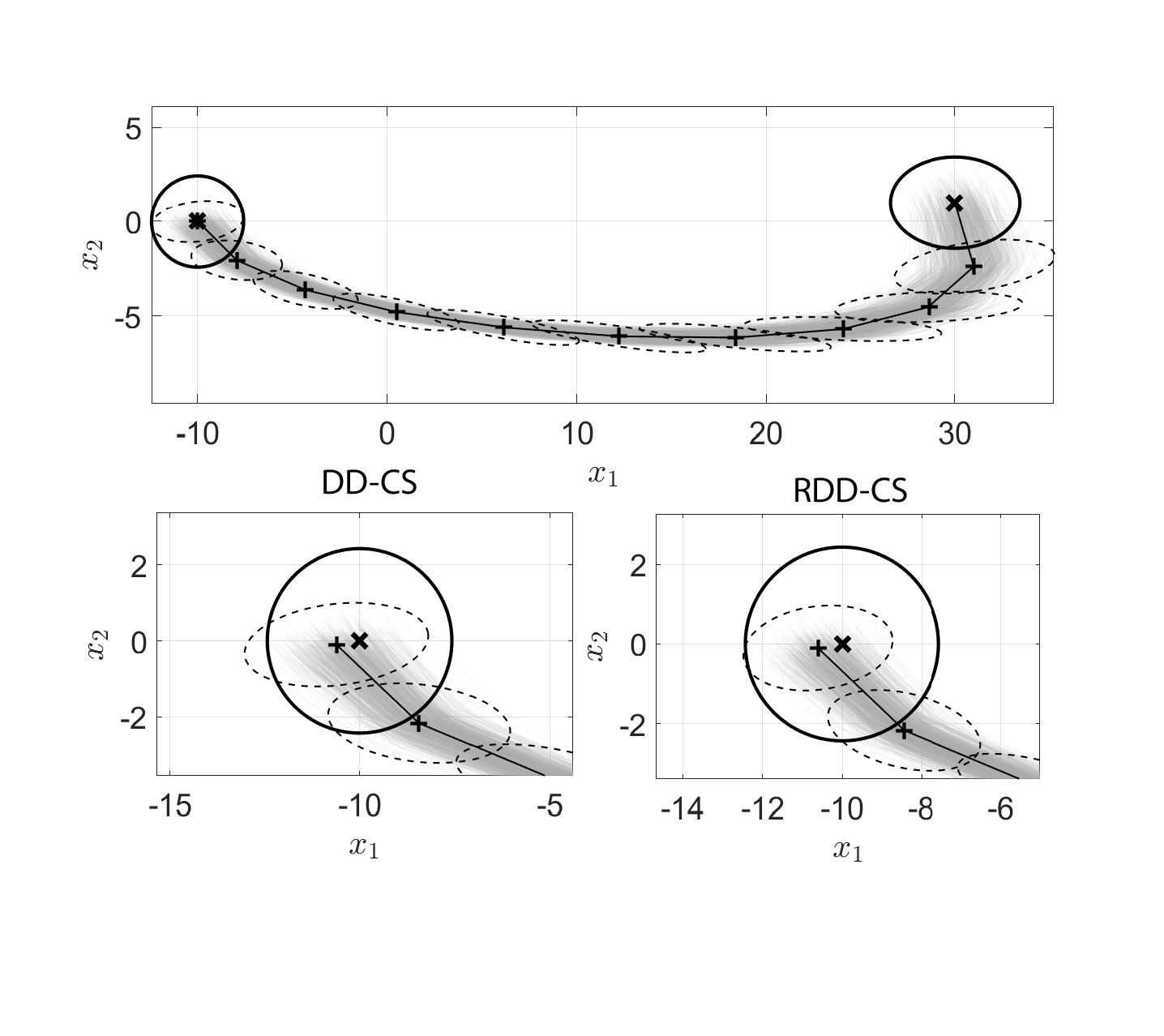}}
	\hspace{0.3cm}
	\subfigure[MB-CS (top) and R/DD-CS (bottom) for perturbed dynamics model $\{A + \Delta A, B + \Delta B\}$.]{%
		\label{fig:1b}%
		\includegraphics[trim={1.5cm 3cm 2cm 2.1cm}, clip, scale=0.36]{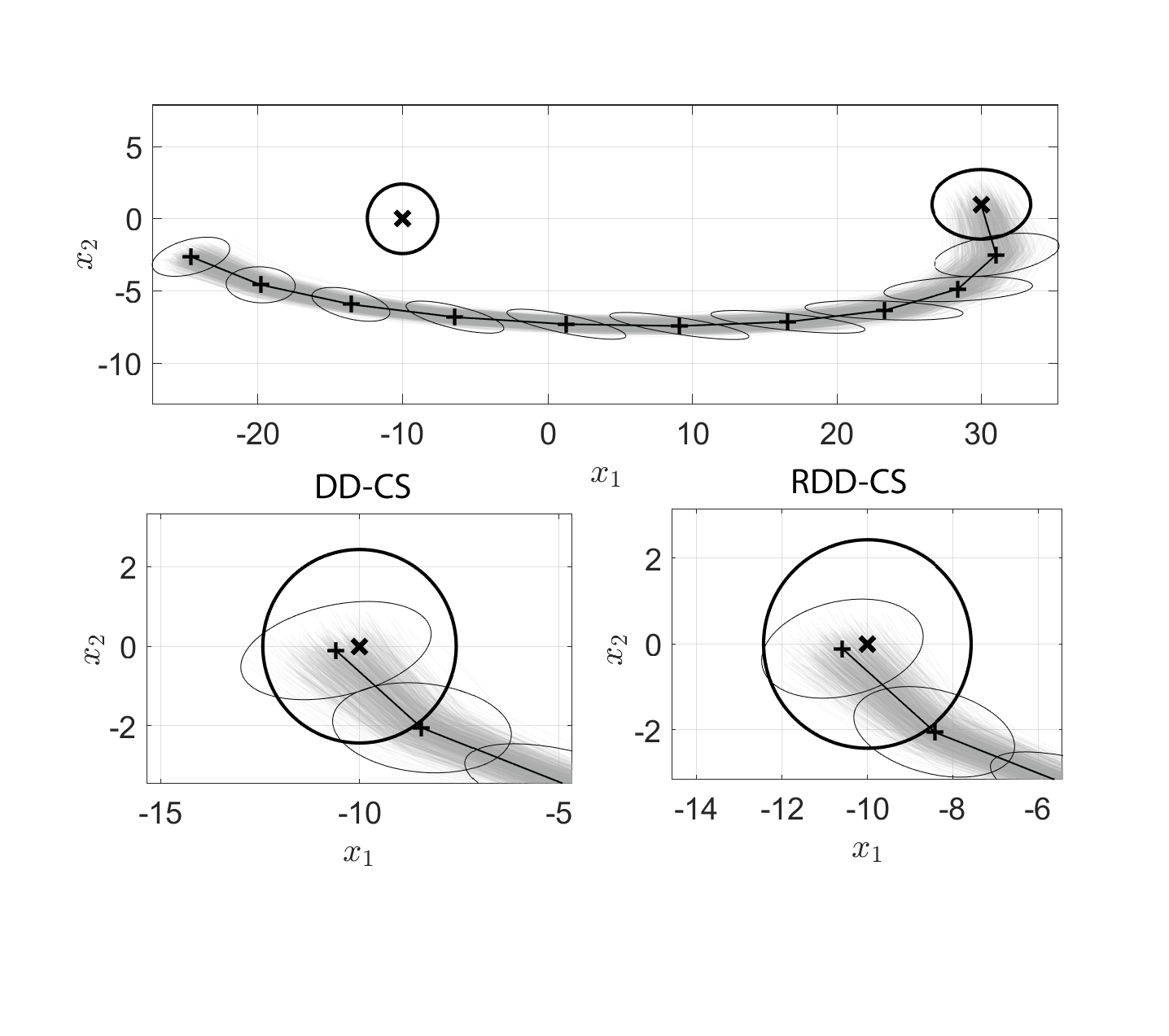}}
	\caption{Comparison of model-based and data-driven covariance steering controllers for (a) perfect system knowledge, and (b) imperfect system knowledge. The robust data-driven solutions are not only adaptable to any linear dynamics due to sampling from the true system but also achieve a desirable terminal covariance.}
	\label{fig:1}
\end{figure}

\section{Numerical Example}~\label{sec:sims}
\vspace*{-0.01cm}
To illustrate the proposed robust DD-CS  (RDD-CS) framework, we consider the ground truth double integrator model
\begin{equation}
	x_{k+1} = 
	\begin{bmatrix}
		1 & \Delta T \\
		0 & 1
	\end{bmatrix} x_{k} + 
	\begin{bmatrix}
		0 \\
		\Delta T
	\end{bmatrix} u_{k} + D w_{k},
\end{equation}
with $\Delta T = 1$ and $D = 0.1 I_{2}$, with initial conditions $\mu_0 = [30, 1]^\intercal, \Sigma_0 = \mathrm{diag}(1, 0.5)$ and terminal conditions $\mu_f = [-10, 0]^\intercal, \Sigma_f = 0.5 I_{2}$.
The planning horizon is chosen to be $N = 10$ and the data-collection horizon $T = 15$ which satisfies the persistence of excitation criterion.
For the uncertainty set, we choose a risk threshold $\delta = 0.001$.
Figure~\ref{fig:1} shows the performance of the model-based and data-driven covariance controllers.
In Figure~\ref{fig:1a}, we see that MB-CS achieves exact uncertainty control, as expected since it has access to the \textit{true} model.
The DD-CS method with MLE noise estimation also achieves similar terminal behavior, although the vanilla DD-CS solution (bottom-left) violates the terminal covariance constraints $\Sigma_{N} \preceq\Sigma_f$, due to the misalignment in the mean steering from noisy data.
%
%
% \begin{wrapfigure}[15]{l}{0.41\textwidth}
% 	\centering
% 	\includegraphics[trim={0 0 0.5cm 0.8cm}, clip, width=0.41\textwidth]{disturbance_estimation_conv.eps}
% 	\vspace*{-0.8cm}
% 	\caption{Convergence of CCP for ML noise estimation DC program.}
% 	\label{fig:2}
% \end{wrapfigure}
%
The robust DD-CS (RDD-CS) solution (bottom-right), on the other hand, achieves a more precise terminal covariance and is fully contained within the desired terminal covariance, even with the slight mean misalignment.
Specifically, the robust solution achieves a terminal covariance
\begin{equation*}
	\Sigma_{N}^{\mathrm{RDD-CS}} = 
	\begin{bmatrix}
		0.2612 & 0.0252 \\
		0.0252 & 0.0941
	\end{bmatrix} \prec \Sigma_f.
\end{equation*}
In Figure~\ref{fig:1b}, we illustrate the performance of the MB and DD controllers in the case where the \textit{true} nominal dynamics model is perturbed with $\Delta A = \tau e_1 e_2^\intercal$ and $\Delta B = \tau e_2$, with $\tau = 0.05$ and $e_i$ is the unit vector along the $i$th axis.
In this case, the model-based design completely fails, as it takes the original model $\{A, B\}$ as the ground truth, while the data-driven design can automatically adapt to any model as it samples data $\mathbb{D}$ from the true underlying system.

\begin{wrapfigure}[19]{r}{0.47\textwidth}
\vspace*{-13pt}
	\centering
	\subfigure[MB-CS solution.]{%
		\label{fig:3a}%
		\includegraphics[trim={1cm 5.5cm 2cm 7cm}, clip, width=0.44\textwidth]{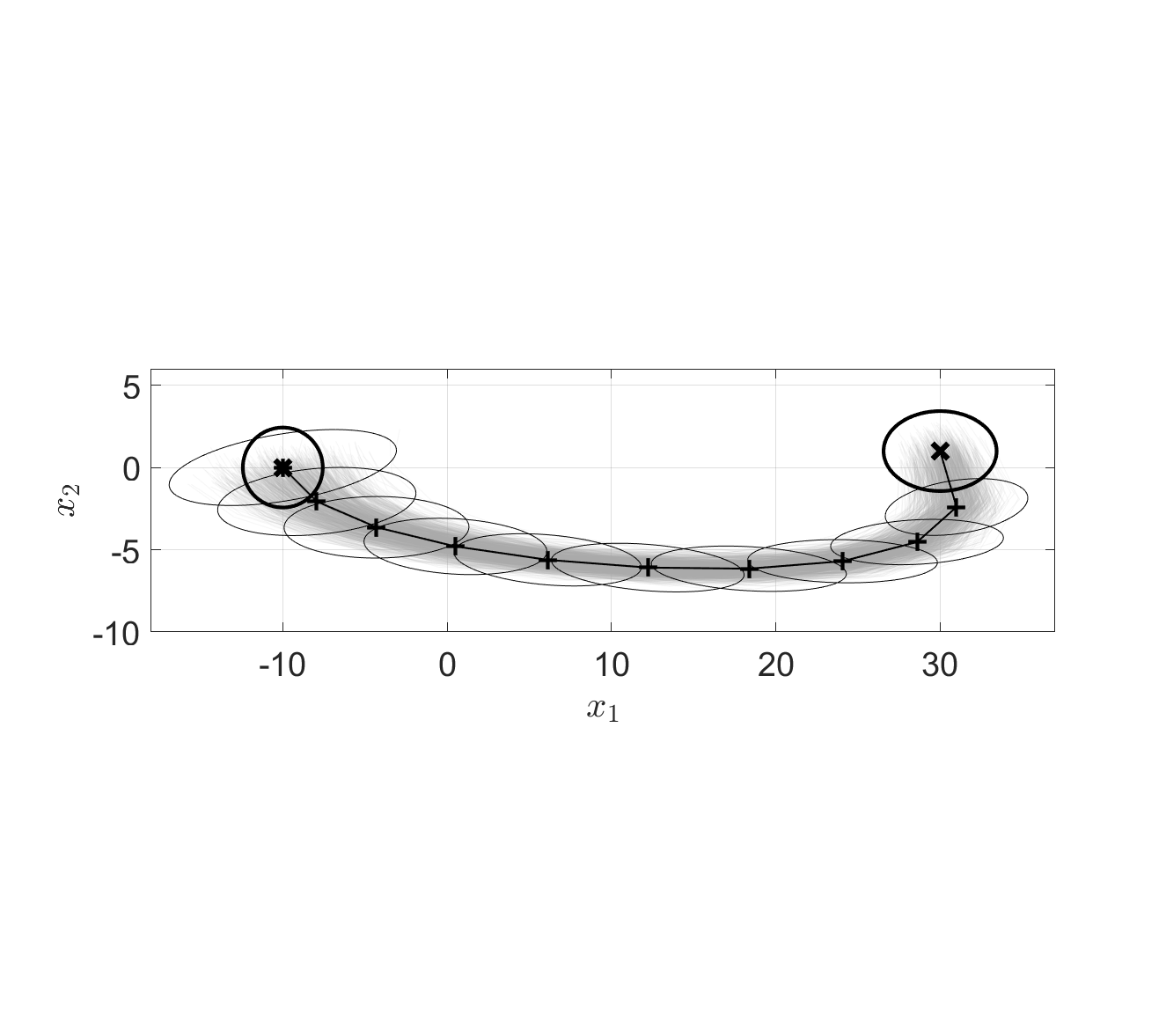}}
	\subfigure[RDD-CS solution.]{%
		\label{fig:3b}%
		\includegraphics[trim={1cm 5.5cm 2cm 6.5cm}, clip, width=0.44\textwidth]{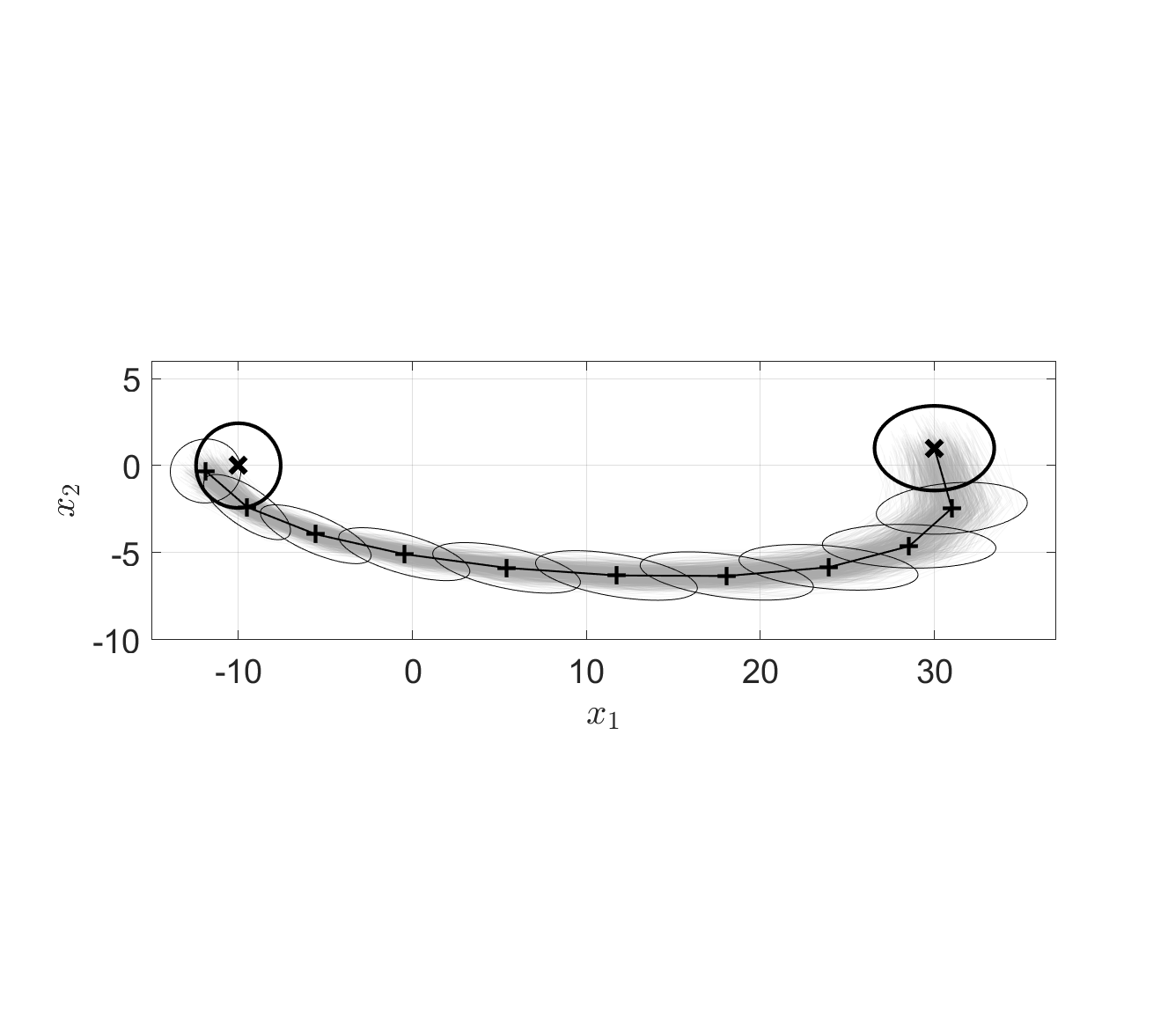}}
	\caption{Model-based (a) and robust data-driven (b) covariance control designs for perturbed disturbance model $D + \Delta D$.
%		The RDD-CS solution achieves the desired terminal covariance for an unknown disturbance intensity using MLE.
}
	\label{fig:3}
\end{wrapfigure}
Next, we include the estimation of the disturbance covariance $\Sigma_{\vct\xi}$ into the DD-CS algorithm.
%Figure~\ref{fig:2} shows the convergence of the sequential convex programming procedure to obtain a local minimum for the program \eqref{eq:ML_problem_DC} for 1,000 trials, yielding an average error of $\E[\|\Delta\Sigma_{\vct\xi}\|] = 0.0014$.
The estimate $\hat{\Sigma}_{\vct\xi}$ is subsequently used in the resulting RDD-CS program.
The results are akin to those of Figure~\ref{fig:1} and hence are omitted.
We then perturb the disturbance matrix from the ground truth model with $\Delta D = 0.2I_{2}$, representing cases in which the modeler cannot effectively quantify the intensity of the disturbance.
Figure~\ref{fig:3} shows that under this new disturbance structure, the model-based design fails as it anticipates weaker disturbances than the actual system experiences, while the data-driven solution achieves the desired terminal covariance, albeit the terminal mean state is more heavily perturbed due to the increased uncertainty.
Designing a robust open-loop control to also take into account the uncertainty in the mean motion is a topic of interest for future work.

%---------------------------------------------------%
\vspace*{-0.07cm}
\section{Conclusion}~\label{sec:conclusion}

We have presented a data-driven uncertainty control method to steer the distribution of an unknown linear dynamical system subject to Gaussian disturbances (DD-CS).
Since the underlying collected data is noisy, an exact system representation is infeasible, and thus we have proposed a maximum likelihood estimation (MLE) scheme to solve for the underlying noise realization from the data as a difference-of-convex program.
Using the statistical properties of MLE, we derived bounds on the $(1-\delta)$-quantile of the uncertainty estimates, which is subsequently used to solve a robust 
DD-CS problem.
Numerical examples show that the robust DD-CS method achieves desirable levels of performance for an underlying unknown linear system compared to that of the model-based counterpart and any disturbance intensity satisfying the theoretically achievable bound~\citep{exact_CS}.
Future work will investigate robust solutions to the mean steering problem in the case of noisy data, as well as extensions to output-feedback systems, moving-horizon implementations, and non-linear systems.

\section{Acknowledgment}

This work has been supported by NASA University Leadership Initiative award 80NSSC20M0163 and ONR award N00014-18-1-2828.
The article solely reflects the opinions and conclusions of its authors and not any NASA entity.

% % % % % % % % % % % % % % % % % % % % 

\bibliography{refs.bib}

\input{appendix}

\end{document}

%% file: appendix.tex
\newcommand{\spec}{\operatorname{spec}}

\title{Technical Appendix}

\section*{A.~Proof of Theorem~1}

\setcounter{equation}{0}
\renewcommand{\theequation}{A.\arabic{equation}}

\begin{proof}
Using \eqref{eq:WFL_gains} we can re-write the covariance dynamics \eqref{eq:covDynamics} as
\begin{align}  \label{eq:covDynamics_Gvars}
    \Sigma_{k+1} &= [B \ \ A]
    \begin{bmatrix}
        K_k \\
        I_n
    \end{bmatrix}
    \Sigma_k 
    \begin{bmatrix}
        K_k \\
        I_n
    \end{bmatrix}^\intercal
    [B \ \ A]^\intercal + DD^\intercal \nonumber \\
    &= (X_{1,T} - \Xi_{0,T}) G_k \Sigma_k G_k^\intercal (X_{1,T} - \Xi_{0,T})^\intercal + \Sigma_{\vct\xi}.
\end{align}
Similarly, the covariance cost in \eqref{eq:covCost} can be re-written as 
\begin{equation}~\label{eq:covCost_Gvars}
    J_{\Sigma, k} = \mathrm{tr}(Q_k\Sigma_k) + \mathrm{tr}(R_k U_{0,T} G_k \Sigma_k G_k^\intercal U_{0,T}^\intercal).
\end{equation}
%To remedy the nonlinearity $G_k\Sigma_k G_k^\intercal$ in %the covariance dynamics and the cost, 
Define the new decision variables $S_k := G_k \Sigma_k\in\Re^{T\times n}$, yielding the covariance dynamics 
\begin{equation}~\label{eq:covDynamics_Svars}
    \Sigma_{k+1} = (X_{1,T} - \Xi_{0,T}) S_k \Sigma_k^{-1} S_k (X_{1,T} - \Xi_{0,T})^\intercal + \Sigma_{\vct\xi},
\end{equation}
and the covariance cost
\begin{equation}~\label{eq:covCost_Svars}
    J_{\Sigma, k} = \mathrm{tr}(Q_k \Sigma_k) + \mathrm{tr}(R_k U_{0,T} S_k \Sigma_k^{-1} S_k^\intercal U_{0,T}^\intercal).
\end{equation}
This problem is still non-convex due to the nonlinear term $S_k\Sigma_k^{-1}S_k^\intercal$.
To this end, relax the covariance dynamics to an inequality constraint and relax the cost by defining a new decision variable $Y_k \succeq S_k \Sigma_k^{-1}S_k^\intercal$.
The relaxed problem is now convex, since the inequality constraints can be written using the Schur complement as the linear matrix inequalities (LMI)
\begin{subequations}
    \begin{align}
        \begin{bmatrix}
            \Sigma_k & S_k^\intercal \\
            S_k & Y_k
        \end{bmatrix} \succeq 0, ~~~~~~~~~~~& \\
        \begin{bmatrix}
            \Sigma_{k+1} - \Sigma_{\vct\xi} & (X_{1,T} - \Xi_{0,T})S_k \\
            S_k^\intercal (X_{1,T} - \Xi_{0,T})^\intercal & \Sigma_{k}
        \end{bmatrix} &\succeq 0.
    \end{align}
\end{subequations}
The cost and the equality constraints, on the other hand, are linear in all the decision variables, and hence convex.
\end{proof}
\section*{B.~Proof of Theorem~3}

\setcounter{equation}{0}
\renewcommand{\theequation}{B.\arabic{equation}}
\begin{proof}
Substituting the multivariable normal statistics PDF into the maximum likelihood cost in \eqref{eq:ML_problem} yields the minimization
\begin{subequations}
    \begin{align*}
        \min_{\Xi_{0,T}, \Sigma_{\vct\xi}} \ \mathcal{J}_{\mathrm{MLE}} &= \min_{\Xi_{0,T}, \Sigma_{\vct \xi}} \ \sum_{k=0}^{T-1} \bigg(-\frac{n}{2}\log(2\pi) - \frac{1}{2}\log\det\Sigma_{\vct\xi} - \frac{1}{2}\xi_k^\intercal \Sigma_{\vct\xi}^{-1} \xi_k\bigg) \nonumber \\
        &= \min_{\Xi_{0,T}, \Sigma_{\vct\xi}} \ -\frac{T}{2}\ \log\det\Sigma_{\vct\xi} - \frac{1}{2}\mathrm{tr} ( \Xi_{0,T}^\intercal \Sigma_{\vct\xi}^{-1}\Xi_{0,T} ) .
    \end{align*}
\end{subequations}
Thus, the MLE problem \eqref{eq:ML_problem} becomes
\begin{subequations} \label{eq:ML_problem_nonconvex}
    \begin{align}
        &\min_{\Xi_{0,T}, \Sigma_{\vct\xi}} \ \left(\frac{T}{2}\log\det\Sigma_{\vct\xi} + \frac{1}{2}\mathrm{tr} ( \Xi_{0,T}^\intercal \Sigma_{\vct\xi}^{-1}  \Xi_{0,T} ) \right) \\
        &\qquad (X_{1,T} - \Xi_{0,T})(I_{T} - S^\dagger S) = 0.
    \end{align}
\end{subequations}
Note that \eqref{eq:ML_problem_nonconvex} is \textit{non}-convex, since $\log\det\Sigma_{\vct\xi}$ is a concave function, and the second term in the objective is nonlinear in the decision variables.
To this end, we can relax this problem by defining the new decision variable $U \succeq \Xi_{0,T}^\intercal\Sigma_{\vct\xi}^{-1}\Xi_{0,T}$ such that the resulting program becomes
\begin{subequations} \label{eq:appA:opt1}
    \begin{align}
        &\min_{\Xi_{0,T}, \Sigma_{\vct\xi}, U} \ \left( \frac{1}{2}\mathrm{tr}(U) + \frac{T}{2}\log\det\Sigma_{\vct\xi} \right)  \\
        &\quad \quad G := (X_{1,T} - \Xi_{0,T})(I_{T} - S^\dagger S) = 0, \\
        &\quad\quad C := \Xi_{0,T}^\intercal \Sigma_{\vct\xi}^{-1} \Xi_{0,T} - U \preceq 0,
    \end{align}
\end{subequations}
where the last matrix inequality may be written as the LMI
\begin{equation}
    \begin{bmatrix}
        \Sigma_{\vct\xi} & \Xi_{0,T} \\
        \Xi_{0,T}^\intercal & U
    \end{bmatrix} \succeq 0.
\end{equation}
The above relaxation is \textit{lossless}, meaning the optimal solution to the DC program aligns with the solution to the original problem.
To see this, construct the Lagrangian of the optimization problem \eqref{eq:appA:opt1} as
\begin{equation*}
    \mathcal{L} = \frac{1}{2}\mathrm{tr} ( U  ) + \frac{T}{2}\log\det\Sigma_{\vct\xi} + \mathrm{tr}(MC) + \mathrm{tr}(\Lambda^\intercal G),
\end{equation*}
where $M=M^\intercal\in\Re^{T\times T} \succeq 0$ and $\Lambda\in\Re^{n\times T}$ are the Lagrange multipliers corresponding to the inequality and equality constraints, respectively.
The first-order necessary conditions require
\begin{equation*}
    \frac{\partial\mathcal{L}}{\partial U} = \frac{1}{2}I_{T} - M = 0 \Rightarrow M = \frac{1}{2}I_{T} \succ 0.
\end{equation*}
Since the Lagrange multipliers corresponding to the inequality constraint are all strictly positive, this implies that the constraints must \textit{all} be active in order to satisfy complementary slackness, that is, if $M \succ 0, C \preceq 0$, then $\mathrm{tr}(MC) = 0 \Rightarrow C \equiv 0$.
\end{proof}

\section*{C.~Proof of Corollary~1}

\setcounter{equation}{0}
\renewcommand{\theequation}{C.\arabic{equation}}

\begin{proof}
Starting from \eqref{eq:ML_problem_nonconvex}, introduce the Lagrangian
\begin{equation*}
    \mathcal{L} = \frac{1}{2}\mathrm{tr} ( \Xi_{0,T}^\intercal \Sigma_{\vct\xi}^{-1}\Xi_{0,T} )  + \mathrm{tr} \left( \Lambda^\intercal (X_{1,T} - \Xi_{0,T})(I_{T} - S^\dagger S) \right) .
\end{equation*}
The first-order necessary conditions yield
\begin{align} \label{eq:minimizer}
    \frac{\partial\mathcal{L}}{\partial\Xi_{0,T}} = &\Sigma_{\vct\xi}^{-1} \Xi_{0,T} - \Lambda (I_{T} - S^\dagger S) = 0
  ~~  \iff  ~~\Xi_{0,T} = \Sigma_{\vct\xi}\Lambda (I_{T} - S^\dagger S), 
\end{align}
where we used the fact that $I_{T} - S^\dagger S$ is symmetric.
Combining \eqref{eq:minimizer} with the constraint $(X_{1,T} - \Xi_{0,T})(I_{T} - S^\dagger S) = 0$ yields that
    $ \big( X_{1,T} - \Sigma_{\vct\xi} \Lambda (I_{T} - S^\dagger S) \big) (I_{T} - S^\dagger S) = 0$ or that 
    $X_{1,T} (I_{T} - S^\dagger S) - \Sigma_{\vct\xi} \Lambda (I_{T} - S^\dagger S) = 0 $ since $I_{T} - S^\dagger S$ is idempotent.
    The last equality is equivalent to
   $\Sigma_{\vct\xi} \Lambda (I_{T} - S^\dagger S) = X_{1,T} (I_{T} - S^\dagger S)$ and, finally, 
    $\Xi_{0,T} = X_{1,T} (I_T - S^\dagger S)$ as claimed.
\end{proof}

\section*{D.~Proof of Lemma~1}

\setcounter{equation}{0}
\renewcommand{\theequation}{D.\arabic{equation}}
\begin{proof}
For the unconstrained MLE problem of estimating the normally distributed parameters $\xi = [\xi_0^\intercal, \ldots, \xi_{T-1}^\intercal]^\intercal$, the FIM is given by $\mathcal{I} = I_{T}\otimes\Sigma_{\vct\xi}$.
The consistency constraints in vectorized form, $(\Gamma \otimes I_{n})\xi - \lambda = 0$, have the gradient $J = \Gamma \otimes I_{n}$.
Hence, the asymptotic covariance of the uncertainty realization estimates is
\begin{align*}
    \Sigma_{\Delta} &= (I_{T} \otimes \Sigma_{\vct\xi}) - (I_{T} \otimes \Sigma_{\vct\xi})(\Gamma\otimes I_{n})[(\Gamma\otimes I_{n})(I_{T}\otimes \Sigma_{\vct\xi})(\Gamma\otimes I_{n})]^{-1}(\Gamma\otimes I_{n})(\Sigma_{\vct\xi}\otimes I_{T}) \\
    &= I_{T} \otimes \Sigma_{\vct\xi} - \Gamma \otimes \Sigma_{\vct\xi} \\
    &= S^\dagger S \otimes \Sigma_{\vct\xi},
\end{align*}
where in the second equality, we use the facts that $\Gamma$ is symmetric and idempotent.
In the last equality, we use the definition  $\Gamma = I_{T} - S^\dagger S$.
\end{proof}

\begin{comment}
\begin{proof}
Since $\Delta\xi_k\sim\mathcal{N}(0,\Sigma_{\Delta})$, the quantity $\Delta\xi_k^\intercal\Sigma_{\Delta}^{-1}\Delta\xi_k$ follows a chi-square distribution with $n$ degrees of freedom, which implies the random variable $\|z\|\sim\rchi_{n}$, where $z:= \Sigma_{\Delta}^{-1/2}\Delta\xi_k$.
Thus, the $(1-\delta)$-quantile of $\|z\|$ is given by $\Delta_{z} = \{\|z\| \leq \rchi_{n,1-\delta}\}$, which implies the $(1-\delta)$-quantile of $\|\Delta\xi_{k}\|$ is given by $\Delta_{\xi} = \{\|\Delta\xi_k\| \leq \|\Sigma_{\Delta}^{1/2}\|\rchi_{n,1-\delta}\}$.
Lastly, since the columns $\Delta\Xi_{0,T}$ are comprised of $\Delta\xi_k, \ k = 0,\ldots, T-1$, it follows that
\begin{equation}
    \|\Delta\Xi_{0,T}\| \leq \sqrt{T} \max_{k=0,\ldots, T-1} \|\Delta\xi_k\|,
\end{equation}
which provides an upper bound on the norm of the Hankel matrix based on the worst-case norm of the individual $\Delta\xi_k$.
Since we have a high-probability uncertainty set for $\Delta\xi_k$, taking the \textit{worst-case} scenario where each $\|\Delta\xi_k\| = \|\Sigma_{\Delta}^{1/2}\|\rchi_{n,1\delta}$ leads to the desired result, which is a conservative over-approximation of the $(1-\delta)$-quantile of $\|\Delta\Xi_{0,T}\|$.
\end{proof}
\end{comment}

\section*{E.~Proof of Proposition~1}

\setcounter{equation}{0}
\renewcommand{\theequation}{E.\arabic{equation}}
\begin{proof}
It is known that the uncertainty set
\begin{equation}
    \Delta_{\xi} := \{\Delta\xi : \Delta\xi^\intercal\Sigma_{\Delta}^{-1}\Delta\xi \leq \rchi_{nT, 1-\delta}^{2}\}, \label{eq:uncertainty_set_xi}
\end{equation}
contains the $(1-\delta)$-quantile of the distribution of $\Delta\xi$~\citep{billingsley_probability_2012}.
To turn \eqref{eq:uncertainty_set_xi} into an uncertainty set for $\Delta\Xi_{0,T} = \vectorize^{-1}(\Delta\xi)$, recall that a quadratic form is bounded by
\begin{equation}
    \lambda_{\min}(\Sigma_{\Delta}^{-1})\|\Delta\xi\|^{2} \leq \Delta\xi^\intercal\Sigma_{\Delta}^{-1}\Delta\xi \leq \lambda_{\max}(\Sigma_{\Delta}^{-1})\|\Delta\xi\|^{2}.
\end{equation}
Thus,$\Delta\xi^\intercal\Sigma_{\Delta}^{-1}\Delta\xi \leq \rchi^{2}_{nT,1-\delta}$ implies that $\lambda_{\min}(\Sigma_{\Delta}^{-1})\|\Delta\xi\|^{2} \leq \rchi^{2}_{nT,1-\delta}$.
Next, note from the definition of the Frobenius norm, that
\begin{equation*}
    \|\Delta\Xi_{0,T}\|_{\rm F}^{2} = \mathrm{tr}(\Delta\Xi_{0,T}^\intercal \Delta\Xi_{0,T}) = \|\Delta\xi\|^{2}.
\end{equation*}
Thus, the uncertainty set \eqref{eq:uncertainty_set_xi} can be overapproximated by the set
\begin{equation}
    \Delta_{\rm F} = \left\{\Delta\Xi_{0,T} : \|\Delta\Xi_{0,T}\|_{\rm F} \leq   \rchi_{nT,1-\delta}/
    \sqrt{ \lambda_{\min}(\Sigma_{\Delta}^{-1})   }\right\}.
\end{equation}
Letting $\rho =   \rchi_{nT,1-\delta} / \sqrt{ \lambda_{\min} (\Sigma_{\Delta}^{-1}) }$, and noting that $\|\Delta\Xi_{0,T}\| \leq \|\Delta\Xi_{0,T}\|_{\rm F}$, we achieve the desired result.
%
%The parameter $\epsilon \geq 1$ is used to tune the degree of over-approximation of the original set $\Delta_{\xi}$, where $\epsilon = 1$ corresponds to the tightest set $\Delta \supset \Delta_{\xi}$.
\end{proof}

\section*{F.~Proof of Corollary~2}

\setcounter{equation}{0}
\renewcommand{\theequation}{F.\arabic{equation}}

\begin{proof}
From Lemma~\ref{lemma:MLE_convergence}, the covariance of the estimation error from the MLE scheme is given as $\Sigma_{\Delta} = S^\dagger S \otimes \Sigma_{\vct\xi}$.
Let
$\Sigma_{\Delta}^{\epsilon} := (S^\dagger S + \epsilon I_{T}) \otimes \Sigma_{\vct\xi}$, 
for some $\epsilon > 0$.
From the properties of the Kronocker product~\citep{kronecker_product}, it follows that
the eigenvalues of the matrix $(\Sigma_{\Delta}^{\epsilon})^{-1}$ are given by
\begin{equation*}
    \spec (\Sigma_{\Delta}^{\epsilon})^{-1} = \left\{\frac{1}{\lambda_i \mu_j}, \ \lambda_i \in\spec(\Sigma_{\vct\xi}), \ \mu_j\in\spec(S^\dagger S + \epsilon I_{T})\right\}.
\end{equation*}
It follows that
\begin{equation*}
    \lambda_{\min} (\Sigma_{\Delta}^{\epsilon})^{-1} =  \frac{1}{\lambda_{\max}(\Sigma_{\vct\xi}) \mu_{\max} (S^\dagger S + \epsilon I_{T})  } = \|\Sigma_{\vct\xi}^{1/2}\|^{-2}(1 + \epsilon)^{-1},
\end{equation*}
where we use the fact that $\lambda_{\max}(\Sigma_{\vct\xi}) = \sigma_{\max}(\Sigma_{\vct\xi}) = \|\Sigma_{\vct\xi}^{1/2}\|^2$ and that $\spec(S^\dagger S + \epsilon) = \{\epsilon, 1 + \epsilon\}$, since $S^\dagger S$ is a projection matrix.
Taking the limit as $\epsilon \rightarrow 0$, we get
\begin{equation*}
    \lim_{\epsilon\rightarrow 0} \lambda_{\min}(\Sigma_{\Delta}^{\epsilon})^{-1} = \lambda_{\min}(\Sigma_{\Delta}^{-1}) = \|\Sigma_{\vct\xi}^{1/2}\|^{-2},
\end{equation*}
and the result follows immediately.
\end{proof}